\documentclass[a4paper,11pt,reqno]{amsart}

\usepackage[utf8]{inputenc}
\usepackage{mathrsfs,dsfont,hyperref,amsaddr,amsmath,amssymb,amsthm,amsfonts,amstext,amsopn,amsxtra,mathrsfs,esint,enumitem,bookmark,mathtools,bbm}
\usepackage[capitalize]{cleveref}
\usepackage{braket}
\usepackage[backend=biber, giveninits=true,maxbibnames=7]{biblatex} 
\newtheorem{lemma}{Lemma}
\newtheorem{theorem}[lemma]{Theorem}

\newtheorem{corollary}[lemma]{Corollary}
\newtheorem{remark}[lemma]{Remark}

\renewcommand\phi{\varphi}

\newcommand{\norm}[1]{ \left\| #1 \right\|}

\renewcommand{\geq}{\geqslant}
\renewcommand{\leq}{\leqslant}
\renewcommand{\hat}{\widehat}
\renewcommand{\tilde}{\widetilde}

\usepackage{color}

\DeclareMathOperator{\sgn}{sgn}

\newcommand{\functional}{\mathcal{F}(\gamma,\alpha,\rho_0)}
\newcommand{\FF}{\mathcal{F}}
\newcommand{\functionalflex}[3]{\mathcal{F}({#1},{#2},{#3})}
\newcommand{\triple}{(\gamma,\alpha,\rhozero)}
\newcommand{\triplekappa}{(\gammakappa,\alphakappa,\rhozerokappa)}

\newcommand{\RR}{\mathbb{R}}
\newcommand{\RRdrei}{\mathbb{R}^3}

\newcommand{\p}{\partial}

\newcommand{\intp}{\int_{\RR^3}dp\,}
\newcommand{\intq}{\int_{\RR^3}dq\,}
\newcommand{\intpq}{\int_{\RRdrei\times\RRdrei}dpdq\,}

\newcommand{\innerprod}[2]{\langle #1, #2\rangle}
\newcommand{\abs}[1]{\left\vert #1 \right\vert}

\newcommand{\ltwonorm}[1]{\norm{#1}_{L^2(\RRdrei)}}

\newcommand{\lp}{L^p(\RRdrei)}
\newcommand{\ltwo}{L^2(\RRdrei)}
\newcommand{\lone}{L^1(\RRdrei)}
\newcommand{\linf}{L^\infty(\RRdrei)}

\newcommand{\Vhat}{\hat{V}}
\newcommand{\Vhatp}{\hat{V}(p)}
\newcommand{\Vhatq}{\hat{V}(q)}
\newcommand{\Vhatpq}{\hat{V}(p-q)}
\newcommand{\gp}{\gamma (p)}
\newcommand{\gq}{\gamma (q)}
\newcommand{\ap}{\alpha (p)}
\newcommand{\aq}{\alpha (q)}
\newcommand{\rhozero}{\rho_0}

\newcommand{\gammakappa}{\gamma_\kappa}
\newcommand{\alphakappa}{\alpha_\kappa}
\newcommand{\rhozerokappa}{\rho_{0,\kappa}}
\newcommand{\rhokappa}{\rho_\kappa}
\newcommand{\rhogammakappa}{\rho_{\gamma,\kappa}}

\newcommand{\gammakappap}{\gamma_\kappa(p)}
\newcommand{\alphakappap}{\alpha_\kappa(p)}
\newcommand{\gammakappaq}{\gamma_\kappa(q)}
\newcommand{\alphakappaq}{\alpha_\kappa(q)}

\date{\today}

\begin{document}
 
\title[On Ground States of the Bogoliubov Energy Functional: A Direct Proof]{On Ground States of the Bogoliubov Energy Functional: A Direct Proof}

\author[J. Oldenburg]{Jakob Oldenburg}
\address{Department of Mathematics, LMU Munich, Theresienstrasse 39, 80333 Munich, Germany} 
\email{\href{mailto:jakob.oldenburg@math.uzh.ch}{jakob.oldenburg@math.uzh.ch}}

\date{\today}

\let\thefootnote\relax\footnote{This article may be downloaded for personal use only. Any other use requires prior permission of the author and AIP Publishing. This article appeared in Journal of Mathematical Physics 62, 071902 (2021) and may be found at \href{https://doi.org/10.1063/5.0054712}{https://doi.org/10.1063/5.0054712}.}

\let\thefootnote\relax\footnote{\textit{Present address: }Institute of Mathematics, University of Zürich, Winterthurerstrasse 190, 8057 Zürich, Switzerland}

\begin{abstract}
% insert abstract here
The Bogoliubov energy functional proposed recently by Napi\'orkowski, Reuvers and Solovej is revisited. We offer a direct proof of the existence of minimizers at zero temperature, which covers a significantly larger class of interaction potentials. The ideas used in this proof also imply that in any ground state, more than half of the particles are inside the Bose-Einstein condensate.
\end{abstract}

\maketitle

\section{Introduction and Main Results}
\setcounter{page}{1}
\pagenumbering{arabic}

Almost 100 years ago, in 1924/25, Bose and Einstein predicted that at low temperatures, bosons should macroscopically occupy the same quantum state and thereby form a condensate \cite{boseprediction, einsteinprediction}. However, their analysis focused on an ideal Bose gas and did not take particle interactions into account. An approximate theory of the interacting Bose gas was proposed by Bogoliubov in 1947 in order to explain the occurrence of superfluidity in liquid helium \cite{BogoliubovPrediction}. The first experimental observation of Bose-Einstein condensates was achieved much later: in 1995, both the groups of Cornell and Wieman and the one of Ketterle successfully created condensates \cite{observationWiemanCornell,observationKetterle}. Since then, much progress has been made regarding the mathematical analysis of the Bose gas.

Recently, Napi\'orkowski, Reuvers and Solovej presented Bogoliubov's approximation in a variational formulation: They considered the free energy of a grand canonical, translation-invariant system of bosons in a box for states that describe a Bose-Einstein condensate coupled with a quasi-free state. By taking the thermodynamic limit, they obtained the Bogoliubov free energy functional \cite{bogoliubovenergy}.
This derivation of the functional seems to go back to a paper by Critchley and Solomon from 1976 \cite{CritchleySolomon}. In their analysis of the functional, Napi\'orkowski, Reuvers and Solovej first show that there are minimizers at any temperature $T$ for sufficiently regular interaction potentials and then derive a formula for the critical temperature and an expansion of the ground state energy in the dilute limit (achieved by fixing the expectation value of the density) which is related to the Lee-Huang-Yang formula \cite{bogoliubovenergycriticaltemp}.

It is interesting that the proof of the existence of minimizers at $T=0$ in \cite{bogoliubovenergy} is indirect: The authors first show the existence of minimizers for positive temperature. Then they use that for $T \to 0$ these minimizers form a minimizing sequence of the functional for $T=0$ and obtain a ground state at zero temperature from this sequence. This proof requires strong conditions on the interaction potentials.

In the present work, we revisit the Bogoliubov free energy functional as defined in \cite{bogoliubovenergy} for temperature $T=0$ and chemical potential $\mu\in\RR$ and offer a direct proof of the existence of minimizers for a larger class of interaction potentials. Since the free energy in this case is just the ground state energy, we refer to it as the Bogoliubov energy functional. For particles interacting via a potential $V$, it is given by

\begin{equation}\label{eqfunctional}
\begin{split}
\functional &\coloneqq \intp p^2 \gp  -\mu \rho +\frac{\Vhat(0)}{2}\rho^2 + \rhozero \intp \Vhatp (\gp +\ap)\\
			&\hphantom{\coloneqq}+ \frac{1}{2}\intpq \Vhatpq (\gp\gq+\ap\aq)
\end{split}
\end{equation}
where
\begin{equation}
\rho\coloneqq \rhozero+\rho_\gamma\coloneqq \rhozero +\intp \gp.
\end{equation}
It is defined on the domain
\begin{equation}
\mathcal{D} \coloneqq \{(\gamma,\alpha,\rhozero)\vert \gamma\in L^1(\RRdrei,(1+p^2)dp), \gp \geq 0, \ap^2\leq \gp^2 +\gp,\rhozero\geq 0\}.
\end{equation}
We absorbed a factor $1/(2\pi)^3$ into the measure on $\RRdrei$ for simplicity of the notation.

The function $\gamma$ describes the one-particle density in momentum space; hence, $\rho_\gamma$ can be interpreted as the density of particles outside the condensate while $\rhozero$ is the density of the condensate. The positivity of $\gamma$ and $\rhozero$ reflects the fact that particle numbers are positive. Pairing in the system is described by the real-valued function $\alpha$. The condition $\alpha^2\leq\gamma^2+\gamma$ stems from the fact that the generalized one-particle density matrix associated with a state is positive.

Our main result is
\begin{theorem}[Minimizers of the Bogoliubov Energy Functional]\label{theoremExistence}
Let the interaction potential $V$ satisfy
\begin{equation}
\label{eqpropertiesV}
V\geq 0, \quad \Vhat \geq 0,  \quad V \not\equiv 0, \quad V \in \lone, \quad \Vhat \in \lone.
\end{equation}
For all chemical potentials $\mu\in\RR$, the functional $\FF$ has a minimizer $(\tilde{\gamma},\tilde{\alpha},\tilde{\rho}_0)\in \mathcal{D}$:
\begin{equation}
\inf_{\triple\in\mathcal{D}}\functional = \functionalflex{\tilde{\gamma}}{\tilde{\alpha}}{\tilde{\rho}_0}.
\end{equation}
\end{theorem}
Note that the conditions satisfied by the interaction potential immediately imply that $V, \Vhat\in C^0(\RRdrei)\cap \linf$ and that both $V$ and $\Vhat$ are symmetric: $\Vhatp = \Vhat (-p)$ and $V(x)=V(-x)$.

Some of the ideas used in the proof of the existence of minimizers yield corresponding conditions for any minimizer at $T=0$. This allows to prove three properties of such minimizers.

\begin{corollary}[Properties of Minimizers]\label{corollaryproperties}
Any minimizer $\triple$ of $\FF$ for $\mu >0$ satisfies the following:
\begin{enumerate}
\item $\alpha (p)^2=\gamma (p) ^2 + \gamma (p)$ for almost all $p\in\RRdrei$
\item $\rhozero > \rho_\gamma$
\item $
\gamma (p) \leq 
\begin{cases}
C & \textrm{ for } \abs{p}\leq P_0\\
C\abs{p}^{-4} & \textrm{ for } \abs{p}\geq P_0
\end{cases}
$
\quad for some constants $C$ and $P_0$.
\end{enumerate}
\end{corollary}

The strategy of using the minimizers for $T>0$ to obtain a minimizer at $T=0$ as presented in \cite{bogoliubovenergy} requires stricter conditions for the interaction potential: in addition to the properties stated in Theorem \ref{theoremExistence}, the proof requires $V$ and $\Vhat$ to be radial as well as that $\Vhat \in C^3(\RRdrei)$ with $\nabla \Vhat\in \ltwo$ and that all derivatives of $\Vhat$ up to third order are bounded.
The key ideas in the new argument are the following: First, we use statements from \cite{bogoliubovenergy} to see that we can restrict the minimization problem to fixed densities. Namely, we decompose the full minimization into a minimization of the function
\begin{equation}
\label{eqdefsmallf}
f(\lambda,\rhozero)\coloneqq\inf_{\substack{\gamma, \alpha\\(\gamma,\alpha,\rhozero)\in \mathcal{D}, \rho_\gamma =\lambda}} \functional
\end{equation}
with respect to the densities. Introducing a cutoff makes the problem a lot easier: For the restricted domains
\begin{equation}\label{eqdefrestricteddomain}
\mathcal{D}_\kappa \coloneqq \left\{(\gamma,\alpha,\rhozero)\in\mathcal{D}\mid \gp\leq \kappa \right\}
\end{equation}
it is rather easy to find minimizers called $\triplekappa$. The bound used here is slightly simpler than the $\kappa/p^2$ cutoff chosen in \cite{bogoliubovenergy} for nonzero temperature. The idea is to send $\kappa$ to infinity and use the fact that $\triplekappa$ forms a minimizing sequence of the full problem to obtain a minimizer for it. This is where we introduce a new approach. We analyze the minimizers on the restricted domain for large $\kappa$; most importantly, we see that they satisfy $\alphakappa^2=\gammakappa^2+\gammakappa$ almost everywhere. This property allows to reduce the problem to a minimization over pure states corresponding to a minimization in $\alpha$ and $\rhozero$. Euler-Lagrange equations for this new minimization problem yield important insights. Modifying the minimizers $\triplekappa$ and comparing the energies enables us to show the uniform bound $\gammakappa \leq C/(1+p^4)$ for a subsequence. We conclude that if $\kappa$ is large enough such that the cutoff is strictly larger than this universal bound, then the minimizer on this restricted domain is in fact a minimizer of the full problem.

\section{Preliminaries}

We begin the proof of the existence of minimizers by recalling some basic statements from \cite{bogoliubovenergy}. 
Only consider $\mu > 0$ in the following as for the other case it is easy to see that the Hamiltonian that yields the functional is positive, and therefore, the vacuum is the minimizer (compare \cite[p. 35]{bogoliubovenergy}).

First, we see a lower bound of the functional. In this proof and in several other calculations, we use the fact that $\alpha^2\leq \gamma^2 +\gamma$ and $\gamma\in\lone$ imply $\alpha\in\lone + \ltwo$ as can be seen by writing $\alpha =\alpha \mathbbm{1}_{\{\gamma \leq 1\}}+ \alpha \mathbbm{1}_{\{\gamma > 1\}}$.
\begin{lemma}[Lower bound]
\label{lemmalowerbound}
There are constants $C, \epsilon >0$ such that
\begin{equation}\label{eqlowerbound}
\functional\geq  \intp p^2 \gp + \epsilon (\rho_0^2 +\rho_\gamma^2) -C
\end{equation}
for any $\triple\in\mathcal{D}$.
\end{lemma}
\begin{proof}
The proof is similar to Lemma 4.1 in \cite{bogoliubovenergy}, but it includes the density of the Bose-Einstein condensate. Note that the first term in \eqref{eqlowerbound} is exactly the kinetic term as given in the functional and that the quadratic terms are positive since $V\geq 0$. By $\alpha^2\leq \gamma^2+ \gamma$, we have $\abs{\alpha}\leq \gamma +\sqrt{\gamma}$, and therefore, we can bound the rest of the terms
\begin{equation}
\begin{split}
			&\rhozero \intp \Vhatp \ap -\mu \rho +\frac{\Vhat(0)}{2}\rho^2\\
			&\geq -\rhozero\Vhat(0)\rho_\gamma -\rhozero\ltwonorm{\Vhat}\sqrt{\rho_\gamma} -\mu \rho +\frac{\Vhat(0)}{2}\rho^2\\
			&=\frac{\Vhat(0)}{2}\left(\rho_0^2 +\rho_\gamma^2\right)  -\mu (\rhozero +\rho_\gamma)  -\rhozero\ltwonorm{\Vhat}\sqrt{\rho_\gamma}\\
			&\geq\frac{\Vhat(0)}{2}\left(\rho_0^2 +\rho_\gamma^2\right) -\mu (\rhozero+\rho_\gamma)  -\frac{\ltwonorm{\Vhat}}{2}\left( K \rho_\gamma + \frac{1}{K}\rhozero^2\right)\\
			&\geq \epsilon (\rho_0^2 +\rho_\gamma^2) - C
\end{split}
\end{equation}
where we choose $K>0$ large enough such that $\rhozero^2$ has a positive prefactor.
This shows that the functional is bounded from below and also that the densities have to be bounded for minimizing sequences.
\end{proof}

\begin{remark}
An easier way to obtain the lower bound is using that $\ap +\gp\geq -1/2$ since $\alpha^2 \leq \gamma^2 +\gamma$ and that $\Vhat\in \lone$. But the calculation here shows that we do not need the condition $\Vhat \in \lone$ for handling the linear term but only for ensuring that the contribution of the $L^2$ part of $\alpha$ to the quadratic term is well-defined. 
\end{remark}

We see that we can restrict the minimization to states with a fixed density.
\begin{lemma}[Convexity and continuity of the functional]
\label{lemmaconvexity}
The functional $\functional$ is jointly convex in $(\gamma,\alpha)$ and $\{\triple\in \mathcal{D}\vert \rho_0=\bar{\rho}_0\}$ is a convex set for every fixed $\bar{\rho}_0\geq 0$. Furthermore, $f(\lambda,\rhozero)$ as defined in equation \eqref{eqdefsmallf} is continuous in both variables, strictly convex in $\lambda$ and goes to infinity if one of the variables goes to infinity, so that
\begin{equation}
\inf_{(\gamma,\alpha,\rhozero)\in \mathcal{D}} \functional = f(\lambda,\rhozero)
\end{equation}
for some $\lambda$, $\rhozero$.
\end{lemma}
\begin{proof}
This was shown in Proposition 4.1 in \cite{bogoliubovenergy} without the term $-\mu \rho$ in the functional which does not change the properties as also mentioned in \cite[Remark 4.4]{bogoliubovenergy}.
\end{proof}

\section{Restricted Problem}
Now we analyze the restricted problem on $\mathcal{D}_\kappa = \left\{(\gamma,\alpha,\rhozero)\in\mathcal{D}\vert \gp\leq \kappa \right\}$. First, observe that the previous lemma also holds in this case.

\begin{lemma}
\label{lemmacontinuityrestricted}
Define
\begin{equation}
f_\kappa(\lambda,\rhozero)\coloneqq\inf_{\substack{\gamma, \alpha\\(\gamma,\alpha,\rhozero)\in \mathcal{D}_\kappa, \rho_\gamma =\lambda}} \functional.
\end{equation}
Then we have
\begin{equation}
\inf_{(\gamma,\alpha,\rhozero)\in \mathcal{D}_\kappa}\functional = f_\kappa(\lambda^\kappa,\rhozero^\kappa)
\end{equation}
for some $\lambda^\kappa$, $\rhozero^\kappa$.
\end{lemma}

We can obtain a very useful upper bound on the energy by considering the trial state
\begin{equation}\label{eqtrialstateupperbpund}
\bar{\gamma}\coloneqq \gamma_0 \mathbbm{1}_{\{\abs{p}\leq \epsilon\}},\quad
\bar{\alpha}\coloneqq -\sqrt{\gamma_0^2+\gamma_0} \mathbbm{1}_{\{\abs{p}\leq \epsilon\}},\quad
\bar{\rho}_0\coloneqq \frac{\mu}{\Vhat (0)}-\gamma_0 \abs{\{\abs{p}\leq \epsilon\}}
\end{equation}
and choosing $\gamma_0$ large and $\epsilon$ small enough (in particular such that $\bar{\rho}_0>0$). Note that the total density is fixed at $\mu/\Vhat(0)$, which gives an energy contribution of $-\mu^2/(2\Vhat(0))$.

\begin{lemma}[Energy upper bound]
\label{lemmaupperboundenergy}
There is a constant $\kappa_0$ such that
\begin{equation}
\inf_{(\gamma,\alpha,\rhozero)\in \mathcal{D}_\kappa}\functional < \frac{-\mu^2}{2\Vhat(0)}
\end{equation}
for all $\kappa \geq \kappa_0$. This bound is also correct for the full minimization problem, i.e.
\begin{equation}
\inf_{(\gamma,\alpha,\rhozero)\in \mathcal{D}}\functional < \frac{-\mu^2}{2\Vhat(0)}.
\end{equation}
\end{lemma}
\begin{proof}
A calculation in \cite{bogoliubovenergy} (p. 39f) shows
\begin{equation}
\functionalflex{\bar{\gamma}}{\bar{\alpha}}{\bar{\rho}_0} < \frac{-\mu^2}{2\Vhat(0)}
\end{equation}
for certain choices of the parameters in equation \eqref{eqtrialstateupperbpund}. It is clear that these trial states are in the restricted domains if the cutoff is large enough.
\end{proof}

\begin{theorem}[Existence of minimizers for the restricted problem]
\label{theoremrestrictedproblem}
There is a minimizer $\triplekappa$ for the restricted problem:
\begin{equation}
\inf_{(\gamma,\alpha,\rhozero)\in \mathcal{D}_\kappa}\functional = \functionalflex{\gammakappa}{\alphakappa}{\rhozerokappa}.
\end{equation}
\end{theorem}
\begin{proof}
The proof is based on Proposition 4.2 in \cite{bogoliubovenergy}, but it is slightly modified due to the different cutoff. First, note that by Lemma \ref{lemmacontinuityrestricted}, we can restrict the minimization problem to a fixed condensate density $\rhozerokappa$. Now take $\gamma_n$ and $\alpha_n$ such that $(\gamma_n,\alpha_n,\rhozerokappa)\in \mathcal{D}_\kappa$ is a minimizing sequence. Lemma \ref{lemmalowerbound} together with any upper bound obtained from a trial state implies the following uniform bound in $n$:
\begin{equation}\label{eqrestricteduniformboundminseq}
\rho_{\gamma, n}\coloneqq \intp \gamma_n(p)\leq C, \quad \intp p^2\gamma_n (p) \leq C.
\end{equation}
The idea is to show boundedness of the minimizing sequence in a reflexive space to get weak convergence of a subsequence. Take $s\in (6/5,2)$. Then we claim that $(\gamma_n,\alpha_n)$ is bounded in $L^s(\RRdrei)\times L^s(\RRdrei)$. Due to the uniform bound of the densities and the cutoff in the $\sup$ norm, $\gamma_n$ is in fact bounded in any $\lp$ space for $p\in[1,\infty]$. Using that $\alpha_n^2\leq \gamma_n^2 +\gamma_n$ implies $\abs{\alpha_n}\leq \sqrt{2}\max (\gamma_n,\sqrt{\gamma_n})$ we see that for $\alpha_n$:
\begin{equation}
\begin{split}
\norm{\alpha_n}_{L^s(\RRdrei)}^s
&= \int_{\{p\mid \gamma_n(p)\leq 1\}} dp\,\abs{\alpha_n(p)}^s + \int_{\{p\mid \gamma_n(p)>1\}} dp\,\abs{\alpha_n(p)}^s\\
&\leq  2^{s/2}\int_{\{p\mid \gamma_n(p)\leq 1\}} dp\,\gamma_n(p)^{s/2} + 2^{s/2} \int_{\{p\mid \gamma_n(p)>1\}} dp\,\gamma_n(p)^s\\
&\leq  2^{s/2}\int_{\{p\mid \gamma_n(p)\leq 1, \abs{p} < 1\}} dp\,\gamma_n(p)^{s/2} +
		2^{s/2}\int_{\{p\mid \gamma_n(p)\leq 1, \abs{p} \geq 1\}} dp\,\gamma_n(p)^{s/2} + C\\
&\leq  2^{s/2}\int_{\{p\mid \abs{p} < 1\}} dp\, +
		2^{s/2}\int_{\{p\mid\abs{p} \geq 1\}} dp\,(p^2\gamma_n(p))^{s/2}\abs{p}^{-s} + C.
\end{split}
\end{equation}
By Hölder's inequality, the second term is bounded by 
\begin{equation}
\norm{(p^2\gamma_n)^{s/2}}_{L^{2/s}(\RRdrei)}\norm{\mathbbm{1}_{\{p\mid\abs{p}\geq 1\}} \abs{p}^{-s}}_{L^{2/(2-s)}(\RRdrei)}.
\end{equation}
Here, the first factor is bounded uniformly due to \eqref{eqrestricteduniformboundminseq}, and the second factor is finite for our choice of $s$. Hence, we have uniform boundedness of $\alpha_n$ in $L^s$. Note that since $\alpha_n^2\leq \gamma_n^2 +\gamma_n\leq C$, $\alpha_n$ is actually bounded in any $L^p$ space with $p>6/5$ by interpolation.
Using the boundedness, we find a subsequence such that
\begin{equation}
(\gamma_n,\alpha_n) \rightharpoonup (\tilde{\gamma},\tilde{\alpha})
\end{equation}
in $L^s(\RRdrei)\times L^s(\RRdrei)$.
Mazur's Lemma states that for every weakly convergent sequence, there are convex combinations of the elements that converge strongly to the same limit (see \cite[p. 61]{Brezis}). Thus, we can replace the $\gamma_n$ and $\alpha_n$ with convex combinations thereof and obtain strong convergence. The new sequence (again named $\gamma_n$ and $\alpha_n$) is still a minimizing sequence by the convexity of $\FF$. Taking another subsequence, we get pointwise convergence almost everywhere; call this sequence again $(\gamma_n,\alpha_n)$. By Fatou's Lemma and pointwise convergence, we get $(\tilde{\gamma},\tilde{\alpha},\rhozerokappa)\in \mathcal{D}_\kappa$. We want to show that this is a minimizer, so we need to prove
\begin{equation}
\lim_{n\to \infty} \functionalflex{\gamma_n}{\alpha_n}{\rhozerokappa}\geq \functionalflex{\tilde{\gamma}}{\tilde{\alpha}}{\rhozerokappa}.
\end{equation} 
We compare the different terms in the energy. Since $(p^2-\mu) \gamma_n(p)$ is bounded below by $-\mu\kappa\mathbbm{1}_{\{\abs{p}\leq\sqrt{\mu}\}}\in L^1(\RRdrei)$, we see that by Fatou's Lemma
\begin{equation}
\liminf_{n\to \infty} \left ( \intp p^2 \gamma_n(p) - \mu \intp \gamma_n(p) \right ) \geq \intp p^2 \tilde{\gamma}(p) - \mu \intp \tilde{\gamma} (p) .
\end{equation}
Fatou's Lemma also implies the desired result for the other density terms. The convergence of the linear terms follows from the strong convergence in $L^s(\RRdrei)$ and Hölder's inequality because $\Vhat\in L^{s/(s-1)}(\RRdrei)$. The quadratic terms also converge because of Hölder's and Young's inequality:
\begin{equation}
\begin{split}
\label{eqrestrictedquadraticterms}
\abs{\innerprod{\gamma_n}{\Vhat\ast\gamma_n}-\innerprod{\tilde{\gamma}}{\Vhat\ast\tilde{\gamma}}}
&=\abs{\innerprod{\gamma_n-\tilde{\gamma}}{\Vhat\ast\gamma_n}-\innerprod{\tilde{\gamma}}{\Vhat\ast(\tilde{\gamma}-\gamma_n)}}\\
&\leq \norm{\gamma_n-\tilde{\gamma}}_{L^s(\RRdrei)} \left( \norm{\Vhat\ast\gamma_n}_{L^{s/(s-1)}(\RRdrei)} + \norm{\Vhat\ast\tilde{\gamma}}_{L^{s/(s-1)}(\RRdrei)}\right)\\
&\leq 2 \norm{\gamma_n-\tilde{\gamma}}_{L^{s}(\RRdrei)} \norm{\Vhat}_{L^{s/(2s-2)}} \sup_m \norm{\gamma_m}_{L^{s}(\RRdrei)} \to 0.
\end{split}
\end{equation}
The same calculation also works for the term that is quadratic in $\alpha$; thus, we have a minimizer.
\end{proof}
Now we analyze the structure of the restricted minimizers. For that define 
\begin{equation}
\rhokappa \coloneqq \rhozerokappa + \rhogammakappa \coloneqq \rhozerokappa +\intp \gammakappap.
\end{equation}

\begin{theorem}[Properties of the minimizers of the restricted problem]\label{theorempropertiesrestricted}
For $\kappa \geq \kappa_0$, the minimizers of the restricted problem satisfy
\begin{equation}
\intp \Vhatp (\gammakappap +\alphakappap) < 0,
\end{equation}
\begin{equation}
\rhozerokappa > 0
\end{equation}
as well as for almost every $p$
\begin{equation}
\label{eqminimizerpurekappa}
\alphakappap^2=\gammakappap^2+\gammakappap
\end{equation}
and
\begin{equation}
\label{eqinequalityalpha}
\alphakappap \frac{\p \FF}{\p \alpha}\triplekappa (p)\leq 0.
\end{equation}
Also, there is a symmetric minimizer satisfying
\begin{equation}
\gammakappap =\gammakappa (-p), \quad \alphakappap =\alphakappa (-p).
\end{equation}
\end{theorem}

\begin{proof}
We can apply an argument used in \cite{bogoliubovenergy} for showing that the minimizers at $T=0$ are pure states to the restricted minimizers: By Lemma \ref{lemmaupperboundenergy}, we have for $\kappa \geq \kappa_0$
\begin{equation}
\begin{split}
\frac{-\mu^2}{2\Vhat(0)} &> \functionalflex{\gammakappa}{\alphakappa}{\rhozerokappa}\\
			&= \intp p^2 \gammakappap + \rhozerokappa \intp \Vhatp (\gammakappap +\alphakappap) \\
			&\hphantom{=}+ \frac{1}{2}\intpq \Vhatpq (\gammakappap\gammakappaq+\alphakappap\alphakappaq)\\
			&\hphantom{=}-\mu \rhokappa +\frac{\Vhat(0)}{2}\rhokappa^2\\
			&\geq \rhozerokappa \intp \Vhatp (\gammakappap +\alphakappap) -\frac{\mu^2}{2\Vhat(0)}
\end{split}
\end{equation}
as both the quadratic terms and the kinetic energy are positive and the density terms are bounded below by the given constant. Thus,
\begin{equation}
\intp \Vhatp (\gammakappap +\alphakappap) < 0
\end{equation}
and
\begin{equation}
\rhozerokappa > 0.
\end{equation}
The fact that the condensate density is nonzero means that the energy can not be lowered by varying it; hence, the minimizer has to satisfy the Euler-Lagrange equation
\begin{equation}
\label{eqeulerlagrangerhozero}
0=\frac{\p \FF}{\p \rhozero}\triplekappa = \Vhat (0)\rhokappa -\mu  +\intp \Vhatp (\gammakappap + \alphakappap).
\end{equation}
Furthermore, the functional derivative in $\gamma$ is given by
\begin{equation}
\begin{split}
\frac{\p \FF}{\p \gamma}\triplekappa (p)
&= p^2 + \Vhat (0)\rhokappa -\mu +\rhozerokappa\Vhatp +\Vhat\ast \gammakappap\\
&= p^2 - \intq \Vhatq (\gammakappaq +\alphakappaq) +\rhozerokappa\Vhatp +\Vhat\ast \gammakappap\\
&> 0
\end{split}
\end{equation}
by the previous equations. This implies that up to a set of measure $0$ we need to have
\begin{equation}
\alphakappap^2=\gammakappap^2+\gammakappap
\end{equation}
as otherwise the energy could be lowered by decreasing $\gamma$ on the set where this does not hold.

The next property can be seen similarly. Suppose equation \eqref{eqinequalityalpha} was wrong: reducing $\abs{\alphakappa}$ on the set where the inequality is not satisfied would decrease the energy, and thus, this set can only have measure zero.

We can assume that the minimizers are symmetric because the symmetry of $\Vhat$ implies that the state defined by
\begin{equation}
\bar{\gamma}_\kappa (p)\coloneqq \gammakappa (-p), \quad \bar{\alpha}_\kappa (p)\coloneqq \alphakappa (-p)
\end{equation}
with the density $\rhozerokappa$ has the same energy as $\triplekappa$. Therefore, the convex combination with weights $1/2$ and $1/2$ can only have lower energy as $\triplekappa$ and it has the desired symmetry.
\end{proof}

\section{Uniform Bound}
In this section, we come to the core of the proof, which is showing the boundedness of the restricted minimizers and analyzing their decay for large $p$.
\begin{theorem}[Uniform bound for the restricted minimizers]
\label{theoremBoundrestricted}
There are constants $P_0 > 0$ and $C>0$ such that up to a subsequence (again called $\gammakappa$), the minimizers of the restricted problem satisfy almost everywhere
\begin{equation}
\gammakappap \leq 
\begin{cases}
C & \textrm{ for } \abs{p}\leq P_0\\
C\abs{p}^{-4} & \textrm{ for } \abs{p}\geq P_0.
\end{cases}
\end{equation}
\end{theorem}

\begin{proof}
The idea is to first consider a combined variation in $\gamma$ and $\alpha$, which leads to a bound on $\gammakappa$ in terms of the functional derivatives of $\FF$ with respect to $\gamma$ and $\alpha$. This bound immediately implies the decay for large $\abs{p}$. For small $\abs{p}$, we modify $\triplekappa$ and compare the energy of the modified state to the one of $\triplekappa$ to obtain the desired bounds.

Step 1: We only consider $\kappa\geq \kappa_0$ in the following, where $\kappa_0$ is the constant from Theorem \ref{theorempropertiesrestricted} so that we can use the properties stated therein. By equation \eqref{eqminimizerpurekappa}, we can write $\gammakappa$ in terms of $\alphakappa$. Define
\begin{equation}
\gamma(\alpha)\coloneqq -\frac{1}{2} +\sqrt{\frac{1}{4}+\alpha^2}.
\end{equation}
By definition, $\alpha^2=\gamma (\alpha)^2+\gamma (\alpha)$. Now consider the functional
\begin{equation}
\tilde{\FF}(\alpha,\rhozero)\coloneqq \functionalflex{\gamma(\alpha)}{\alpha}{\rhozero}
\end{equation}
on the domain
\begin{equation}
\mathcal{D}_\kappa^{\alpha} \coloneqq \left\{ (\alpha,\rhozero) \vert \gamma(\alpha)\in L^1(\RRdrei,(1+p^2)dp), \gamma(\alpha)\leq\kappa, \rhozero\geq 0\right\}.
\end{equation}
Clearly, $(\gamma (\alpha), \alpha, \rhozero) \in \mathcal{D}_\kappa$ for all $(\alpha,\rhozero)\in \mathcal{D}_\kappa^{\alpha}$. This implies that $(\alphakappa,\rhozerokappa)$ is a minimizer of the minimization problem on $\mathcal{D}_\kappa^\alpha$. Now consider the variation in $\alpha$ for this functional. First, define
\begin{equation}\label{eqAkappa}
\begin{split}
A_\kappa(p) & \coloneqq \frac{\p \FF}{\p \gamma}\triplekappa (p)\\
&= p^2 + \Vhat (0)\rhokappa -\mu +\rhozerokappa\Vhatp +\Vhat\ast \gammakappap\\
&= p^2 - \intq \Vhatq (\gammakappaq +\alphakappaq) +\rhozerokappa\Vhatp +\Vhat\ast \gammakappap >0
\end{split}
\end{equation}
and
\begin{equation}\label{eqBkappa}
B_\kappa(p) \coloneqq \frac{\p \FF}{\p \alpha}\triplekappa (p)= \rhozerokappa\Vhatp +\Vhat\ast \alphakappap 
\end{equation}
where we used the Euler-Lagrange equation in $\rhozero$ \eqref{eqeulerlagrangerhozero}.
Observe that
\begin{equation}
\frac{\p \gamma(\alpha)}{\p\alpha}=\frac{\alpha}{\sqrt{\frac{1}{4}+\alpha^2}}
\end{equation}
which is bounded by $1$ in absolute value. Thus, we can apply the chain rule to obtain
\begin{equation}
\frac{\p \tilde{\FF}}{\p \alpha} (\alphakappa,\rhozerokappa) = A_\kappa \left.\frac{\p \gamma(\alpha)}{\p\alpha}\right\vert_{\alphakappa}   + B_\kappa
=\frac{\alphakappa}{\sqrt{\frac{1}{4}+\alphakappa^2}} A_\kappa + B_\kappa .
\end{equation}
We need to have
\begin{equation}
\label{eqcombinedeulerlagrange}
\frac{\p \tilde{\FF}}{\p \alpha} (\alphakappa,\rhozerokappa) (p)
\begin{cases}
\leq 0 & \textrm{ for } \alphakappap> 0\\
= 0 & \textrm{ for } \alphakappap= 0\\
\geq 0 & \textrm{ for } \alphakappap< 0
\end{cases}
\end{equation}
for almost all $p$ as otherwise the energy could be lowered by varying $\alpha$. Clearly, $B_\kappa(p)=0$ for $\alphakappap=0=\gammakappap$. But the reverse is also true as $B_\kappa (p)=0$ yields
\begin{equation}
0\geq \frac{\abs{\alphakappap}}{\sqrt{\frac{1}{4}+\alphakappap^2}} A_\kappa (p) \Rightarrow \alphakappap = 0
\end{equation}
as $A_\kappa(p)>0$.
Hence, up to null sets
\begin{equation}
\label{eqzerogammazeroB}
\{p\mid\gammakappap=0\}=\{p\mid B_\kappa(p)=0\}.
\end{equation}
Next, for almost every $p$ we either have $\gammakappap=0=\alphakappap$, which clearly satisfies the desired inequality, or we have $\alphakappap\neq 0$ and
\begin{equation}
%\begin{split}
0 \geq \sgn(\alphakappap) \frac{\p \tilde{\FF}}{\p \alpha} (\alphakappa,\rhozerokappa) (p)
= - \abs{B_\kappa(p)}+\frac{\abs{\alphakappap}}{\sqrt{\frac{1}{4}+\alphakappap^2}} A_\kappa (p)
%\end{split}
\end{equation}
where we used that $\sgn (\alphakappap)B_\kappa(p)=-\abs{B_\kappa(p)}$ for such $p$ by \eqref{eqinequalityalpha}.
Using the strict positivity of $A_\kappa$ and that $B_\kappa (p)\neq 0$ due to \eqref{eqzerogammazeroB}, we see
\begin{equation}
B_\kappa(p)^2\ \geq \frac{1}{1+\frac{1}{4\alphakappap^2}}A_\kappa(p)^2
\end{equation}
implying
\begin{equation}
\label{eqinequalitygammaAB}
\frac{1}{4\alphakappap^2}\geq\frac{A_\kappa(p)^2-B_\kappa (p)^2}{B_\kappa (p)^2}.
\end{equation}
Therefore, finding upper bounds on $\gammakappap$ and $\alphakappap$, respectively, corresponds to bounding $A_\kappa(p)^2-B_\kappa (p)^2$ from below by a positive function.

Step 2: As in the proof of Proposition 5.2 in \cite{bogoliubovenergy}, we can show that up to a subsequence we have the convergence $\rhozerokappa\to \tilde{\rho}_0$ and $\rho_{\gamma,\kappa} \to \tilde{\rho}_\gamma$: By the trial state argument in Lemma \ref{lemmaupperboundenergy} and the lower bound in Lemma \ref{lemmalowerbound}, we see that for $\kappa\geq \kappa_0$ we have
\begin{equation}\label{eqdensitiesbounded}
\begin{split}
\frac{-\mu^2}{2\Vhat(0)} > \functionalflex{\gammakappa}{\alphakappa}{\rhozerokappa}
			\geq -C+ \epsilon (\rhozerokappa^2 + \rhogammakappa ^2).
\end{split}
\end{equation}
This implies boundedness of both densities, and therefore, there is a subsequence for which the densities converge. By replacing the original sequence with this subsequence (again labelled by $\kappa$), we can assume the convergence as desired. For this subsequence, we know that all terms in the Euler-Lagrange equation \eqref{eqeulerlagrangerhozero} except the integral term have a limit; therefore, it converges as well. Using again the trial state argument in Lemma \ref{lemmaupperboundenergy}, we see that
\begin{equation}
\begin{split}
\frac{-\mu^2}{2\Vhat(0)} &> \lim_{\kappa\to\infty} \functionalflex{\gammakappa}{\alphakappa}{\rhozerokappa}\\
			&= \lim_{\kappa\to\infty} \left( \intp p^2 \gammakappap + \rhozerokappa \intp \Vhatp (\gammakappap +\alphakappap)\right.\\
			&\hphantom{= \lim_{\kappa\to\infty} \left(\right.}+ \frac{1}{2}\intpq \Vhatpq (\gammakappap\gammakappaq+\alphakappap\alphakappaq)\\
			&\hphantom{= \lim_{\kappa\to\infty} \left(\right.} \left.-\mu \rhokappa +\frac{\Vhat(0)}{2}\rhokappa^2\right)\\
			&\geq \tilde{\rho}_0  \lim_{\kappa\to\infty} \intp \Vhatp (\gammakappap +\alphakappap) -\mu \tilde{\rho} +\frac{\Vhat(0)}{2}\tilde{\rho}^2\\
			&\geq \tilde{\rho}_0  \lim_{\kappa\to\infty} \intp \Vhatp (\gammakappap +\alphakappap) -\frac{\mu^2}{2\Vhat(0)},
\end{split}
\end{equation}
where $\tilde{\rho}\coloneqq \tilde{\rho}_0+\tilde{\rho}_\gamma$. Thus,
\begin{equation}
\label{eqlimitlinearterm1}
\tilde{\rho}_0>0
\end{equation}
and
\begin{equation}\label{eqlimitlinearterm2}
\lim_{\kappa\to\infty} \intp \Vhatp (\gammakappap +\alphakappap) =-C < 0.
\end{equation}
This also implies 
$
\tilde{\rho}_\gamma>0
$
because the integral term would go to $0$ otherwise as its absolute value is bounded by $C(\rhogammakappa+ \sqrt{\rhogammakappa})$. 
 
Step 3: Now we prove the easier bound for large $p$. Using the previous considerations and the boundedness of $\Vhat$ in the $L^2$ and $L^\infty$ norm, we can conclude that all the terms in $A_\kappa (p)$ and $B_\kappa (p)$ except for the $p^2$ term are uniformly bounded in both $\kappa$ and $p$. Hence, there is a $P_0$ such that $A_\kappa(p)^2-B_\kappa(p)^2 \geq C p^4$ for $\abs{p}\geq P_0$, where $C$ is a strictly positive constant independent of $\kappa$. Plugging this into \eqref{eqinequalitygammaAB} and using again the uniform boundedness of $B_\kappa (p)$, we see that for such $p$
\begin{equation}
\label{eqboundlargep}
\frac{1}{\alphakappap^2}\geq\frac{C p^4}{B_\kappa(p)^2} \Rightarrow \gammakappap \leq \gammakappap +\gammakappap^2 = \alphakappap^2 \leq C p^{-4}.
\end{equation}

Step 4: Next, we consider small $p$. First, we want to show that up to a subsequence $A_\kappa$ and $B_\kappa$ converge uniformly to continuous functions on the ball $\{p\mid\abs{p}\leq P_0\}$. The advantage of considering convergence only on this compact set is that continuity of $\Vhat$ suffices to prove it.
Clearly, by the specific form given by \eqref{eqAkappa} and \eqref{eqBkappa}, it is enough to show the uniform convergence for the convolution terms. 
In order to do so, we check that they are uniformly equicontinuous on the compact set $\{p\mid\abs{p}\leq P_0\}$.
Consider first $\Vhat\ast \gammakappa$. Take $\epsilon >0$. We want to show that there is a $\delta>0$ independent of $\kappa$ such that $\abs{\Vhat\ast \gammakappa (r) - \Vhat\ast \gammakappa (s)}<\epsilon$ for $\abs{r-s}< \delta$ and $\abs{r},\abs{s} \leq P_0$.

\begin{equation}
\begin{split}
&\abs{\Vhat\ast \gammakappa (r) - \Vhat\ast \gammakappa (s)} \leq \intp \abs{\Vhat (r-p)-\Vhat (s-p)}\gammakappap\\
&= \int_{\abs{p}\leq R}dp\, \abs{\Vhat (r-p)-\Vhat (s-p)}\gammakappap + \int_{\abs{p}> R}dp\, \abs{\Vhat (r-p)-\Vhat (s-p)}\gammakappap\\
&\leq \int_{\abs{p}\leq R}dp\, \abs{\Vhat (r-p)-\Vhat (s-p)}\gammakappap + 2\Vhat(0)\int_{\abs{p}> R}dp\, \frac{C}{p^4}\\
&\leq \int_{\abs{p}\leq R}dp\, \abs{\Vhat (r-p)-\Vhat (s-p)}\gammakappap + \frac{\epsilon}{2}\\
&\leq \sup_{\substack{\abs{p}\leq R, \,\abs{r},\abs{s}\leq P_0\\ \abs{r-s}\leq \delta}} \abs{\Vhat (r-p)-\Vhat (s-p)} \int_{\abs{p}\leq R}dp\, \gammakappap +\frac{\epsilon}{2} < \epsilon
\end{split}
\end{equation}
Here we used the bound \eqref{eqboundlargep} on large $p$ and choose $R>P_0$ large enough such that the second term is less than $\epsilon/2$. Recall the uniform bound $\rhogammakappa\leq C$ from step 2.
The Heine-Cantor theorem states that continuous functions on compact sets are uniformly continuous. Therefore, $\Vhat$ is uniformly continuous on the compact set $\{p\mid\abs{p}\leq P_0+R\}$ which allows us to make the supremum small by choosing $\delta$ small enough. Thus, we can bound the first term by $\epsilon/2$.

We can deal with the convolution of $\alphakappa$ in a similar way. Note that $\alphakappa^2+1/4=(\gammakappa+1/2)^2$ implies $\abs{\alphakappa}\leq \gammakappa +1/2$.
\begin{equation}
\begin{split}
&\abs{\Vhat\ast \alphakappa (r) - \Vhat\ast \alphakappa (s)} \leq \intp \abs{\Vhat (r-p)-\Vhat (s-p)}\abs{\alphakappap}\\
&= \int_{\abs{p}\leq R}dp\, \abs{\Vhat (r-p)-\Vhat (s-p)}\abs{\alphakappap} + \int_{\abs{p}> R}dp\, \abs{\Vhat (r-p)-\Vhat (s-p)}\abs{\alphakappap}\\
&\leq \int_{\abs{p}\leq R}dp\, \abs{\Vhat (r-p)-\Vhat (s-p)}\left(\gammakappap +\frac{1}{2}\right) + 2\ltwonorm{\Vhat}\ltwonorm{\alphakappa \mathbbm{1}_{\{\abs{p}> R\}}(p)}\\
&\leq \int_{\abs{p}\leq R}dp\, \abs{\Vhat (r-p)-\Vhat (s-p)}\left(\gammakappap +\frac{1}{2}\right) + 2\ltwonorm{\Vhat}\int_{\abs{p}> R}dp\, \frac{C}{p^4}\\
&\leq \sup_{\substack{\abs{p}\leq R, \,\abs{r},\abs{s}\leq P_0\\ \abs{r-s}\leq \delta}} \abs{\Vhat (r-p)-\Vhat (s-p)} \int_{\abs{p}\leq R}dp\, \left(\gammakappap +\frac{1}{2}\right) + \frac{\epsilon}{2} < \epsilon,
\end{split}
\end{equation}
where we again used the bound \eqref{eqboundlargep} to make the second integral small by choosing $R>P_0$ large. Due to uniform continuity of $\Vhat$ on $\{\abs{p}\leq P_0+R\}$ and $\int_{\abs{p}\leq R}dp\, (\gammakappap+1/2)\leq C$, the first term is small for sufficiently small $\delta$.

We can now use the Arzelà-Ascoli theorem (see \cite[Theorem I.28]{ReedSimon}): $\Vhat \ast \gammakappa$ is uniformly bounded and uniformly equicontinuous on the ball of radius $P_0$, which is a compact set. Thus, there is a subsequence that converges uniformly for $\abs{p}\leq P_0$. Repeating the same argument for $\Vhat \ast \alphakappa$, we obtain uniform convergence for it as well. Note that since all convolution terms here are continuous on the ball, the limiting functions are continuous as well by Theorem I.25 in \cite{ReedSimon}. We again label the chosen subsequence by the index $\kappa$. 
Hence, we have
\begin{equation}
a(p)\coloneqq \lim_{\kappa\to\infty} A_\kappa(p)
\end{equation}
and
\begin{equation}
b(p)\coloneqq \lim_{\kappa\to\infty} B_\kappa(p)
\end{equation}
for $\abs{p}\leq P_0$. The convergence is uniform, and both $a$ and $b$ are continuous on the ball.

Step 5: Equation \eqref{eqinequalitygammaAB} yields the desired bound for small $\abs{p}$ if $A_\kappa\pm B_\kappa \geq \delta$ on the ball with radius $P_0$ for some $\delta >0$ independent of $\kappa$. The idea is to prove that $a^2-b^2\geq\epsilon>0$ on this ball and then use uniform convergence. This inequality can be seen in two steps. First, we consider small $\abs{p}$ and then the rest. Observe that we can write $A_\kappa^2-B_\kappa^2= (A_\kappa-B_\kappa)(A_\kappa + B_\kappa)$. By \eqref{eqAkappa} and \eqref{eqBkappa}, we have
\begin{equation}\label{eqAplusBkappa}
A_\kappa(p) +B_\kappa(p) = p^2 +\Vhat \ast( \gammakappa +\alphakappa)(p) - \intq \Vhatq (\gammakappaq +\alphakappaq) + 2\rhozerokappa\Vhatp  
\end{equation}
and
\begin{equation} \label{eqAminusBkappa}
A_\kappa(p) - B_\kappa(p) = p^2 +\Vhat \ast (\gammakappa -\alphakappa)(p) - \intq \Vhatq (\gammakappaq +\alphakappaq).
\end{equation}
Evaluating these expressions at $p=0$ yields
\begin{equation}
a(0)^2-b(0)^2=-4\tilde{\rho}_0 \Vhat (0) \lim_{\kappa\to \infty} \intp \Vhatp \alphakappap > 0
\end{equation}
due to \eqref{eqlimitlinearterm1} and \eqref{eqlimitlinearterm2}. The limit has to exist by the choice of the subsequence. By continuity, we can find $P_1>0$ such that $a(p)^2-b(p)^2\geq \epsilon >0$ for all $\abs{p}\leq P_1$.

Step 6: Now we turn to the case $P_1\leq \abs{p}\leq P_0$. The key is that we can modify the minimizing sequence and that this cannot decrease the limit of the energy.
Define for $P_1\leq\abs{p}\leq P_0$:
\begin{equation}
\bar{\gamma}_\kappa(q)\coloneqq \gammakappa (p-q),\quad \bar{\alpha}_\kappa(q)\coloneqq \pm\alphakappa (p-q).
\end{equation}
We want to compare the energies: clearly, the shift and reflection do not change the densities, and also, the quadratic terms remain unchanged since $\Vhat$ is symmetric. For the kinetic term we can use that $\gammakappa$ is symmetric to see
\begin{equation}
\intq q^2 \bar{\gamma}_\kappa(q) =\intq q^2 \gammakappa (p-q) = \intq (p-q)^2 \gammakappaq = \intq q^2\gammakappaq + p^2 \rhogammakappa
\end{equation}
where the mixed term is $0$ by symmetry.
Note that
\begin{equation}
\intq \Vhatq (\bar{\gamma}_\kappa(q)+\bar{\alpha}_\kappa(q)) = \Vhat \ast (\gammakappa \pm \alphakappa)(p).
\end{equation}
Since
\begin{equation}
\liminf_{\kappa\to\infty} \functionalflex{\bar{\gamma}_\kappa}{\bar{\alpha}_\kappa}{\rhozerokappa}
\geq \inf_{(\gamma,\alpha,\rhozero)\in \mathcal{D}}\functional = \lim_{\kappa \to\infty} \functionalflex{\gammakappa}{\alphakappa}{\rhozerokappa} 
\end{equation}
we see
\begin{equation}
\label{eqtrialstateshift}
\begin{split}
0&\leq \liminf_{\kappa\to\infty} \left( \functionalflex{\bar{\gamma}_\kappa}{\bar{\alpha}_\kappa}{\rhozerokappa} - \functionalflex{\gammakappa}{\alphakappa}{\rhozerokappa} \right)\\
&=\lim_{\kappa\to\infty} \left( p^2 \rhogammakappa +\rhozerokappa \left( \Vhat \ast (\gammakappa \pm \alphakappa)(p) -  \intq \Vhatq (\gammakappaq + \alphakappaq) \right) \right)\\
&= p^2 \tilde{\rho}_\gamma +\tilde{\rho}_0 \lim_{\kappa\to\infty} \left( \Vhat \ast (\gammakappa \pm \alphakappa)(p) -  \intq \Vhatq (\gammakappaq + \alphakappaq) \right).
\end{split}
\end{equation}
We used that all terms converge by our choice of the subsequence to replace the $\liminf$ with a $\lim$.
If the densities were not there, these terms would equal $a(p) \pm b(p)$ up to the positive term $2\tilde{\rho}_0\Vhat$ in $a+b$.
Hence, we need to show that $\tilde{\rho}_0>\tilde{\rho}_\gamma$ to use \eqref{eqtrialstateshift} for proving the desired inequality for $a(p)\pm b(p)$.
This inequality for the densities follows from another trial state argument. Modify $\triplekappa$ by transferring some particles from $\gammakappa$ into the condensate while keeping the total density fixed: For $\eta\in (0,1)$, consider the states
\begin{equation}
\bar{\gamma}_\kappa(p)\coloneqq \eta \gammakappa (p),
\quad \bar{\alpha}_\kappa(p)\coloneqq \eta\alphakappa (p),
\quad
 \bar{\rho}_{0,\kappa}\coloneqq \rhozerokappa + (1-\eta)\rhogammakappa.
\end{equation}
Clearly, they satisfy the condition $\bar{\alpha}_\kappa^2\leq \bar{\gamma}_\kappa^2 +\bar{\gamma}_\kappa$. Similar to before, the $\liminf$ of the difference of the energies has to be positive, which implies
\begin{equation}
\begin{split}
0&\leq \frac{1}{1-\eta} \liminf_{\kappa\to\infty} \left( \functionalflex{\bar{\gamma}_\kappa}{\bar{\alpha}_\kappa}{\bar{\rho}_{0,\kappa}} - \functionalflex{\gammakappa}{\alphakappa}{\rhozerokappa} \right)\\
&= \frac{1}{1-\eta}\liminf_{\kappa\to\infty}  \left((\eta-1) \intp p^2\gammakappap \right.\\
&\hphantom{=\frac{1}{1-\eta}\liminf_{\kappa\to\infty}  \left(\right.} + (\eta^2-1) \frac{1}{2}\intpq \Vhatpq (\gammakappap\gammakappaq+\alphakappap\alphakappaq)\\
&\hphantom{=\frac{1}{1-\eta}\liminf_{\kappa\to\infty}  \left(\right.} + \left.(\lbrack \rhozerokappa  +(1-\eta)\rhogammakappa)\rbrack \eta -\rhozerokappa ) \intp\Vhatp (\gammakappap+\alphakappap)\right)\\
&=\liminf_{\kappa\to\infty}  \left(- \intp p^2\gammakappap \right.\\
&\hphantom{=\liminf_{\kappa\to\infty}  \left(\right.} -(1+\eta) \frac{1}{2}\intpq \Vhatpq (\gammakappap\gammakappaq+\alphakappap\alphakappaq)\\
& \hphantom{=\liminf_{\kappa\to\infty}  \left(\right.} + \left.( -\rhozerokappa  +\eta\rhogammakappa ) \intp\Vhatp (\gammakappap+\alphakappap)\right)\\
&\leq - \limsup_{\kappa\to\infty} \left(\intp p^2\gammakappap + \frac{1}{2}\intpq \Vhatpq \gammakappap\gammakappaq \right)\\
&\hphantom{\leq } + ( \tilde{\rho}_{0,\kappa} - \eta\tilde{\rho}_{\gamma,\kappa} ) \left(-\lim_{\kappa\to\infty} \intp \Vhatp (\gammakappap +\alphakappap)\right).
\end{split}
\end{equation}
Sending $\eta$ to $1$, we see
\begin{equation}
\label{eqinequalityrhozerorhogamma}
\tilde{\rho}_{0,\kappa}\geq \tilde{\rho}_{\gamma,\kappa} + \frac{\limsup_{\kappa\to\infty} \left(\intp p^2\gammakappap + \frac{1}{2}\intpq \Vhatpq \gammakappap\gammakappaq\right)}{-\lim_{\kappa\to\infty} \intp \Vhatp (\gammakappap +\alphakappap)}.
\end{equation}
The fraction is strictly larger than $0$ because we can find $\epsilon>0$ such that $\Vhatp \geq \Vhat(0)/2$ if $\abs{p}\leq 2\epsilon$ and then
\begin{equation}\label{eqpositivitykineticconvolution}
\begin{split}
&\intp p^2\gamma (p) + \frac{1}{2}\intpq \Vhatpq \gamma (p) \gamma (q)\\
&\geq \epsilon^2 \int_{\abs{p}> \epsilon} dp\,\gamma(p) +\frac{\Vhat(0)}{4} \int_{\abs{p}\leq \epsilon} dp\,\int_{\abs{q}\leq \epsilon} dq\, \gamma(p)\gamma (q)\\
& \geq\inf_{0\leq x\leq \rho_\gamma} \epsilon^2 x + \frac{\Vhat(0)}{4} (\rho_\gamma -x)^2
\end{split}
\end{equation}
implying positivity of the $\limsup$ since $\rhogammakappa\to \tilde{\rho}_\gamma >0 $. Combining $\tilde{\rho}_0>\tilde{\rho}_\gamma$ with \eqref{eqtrialstateshift} we obtain
\begin{equation}
\label{eqineqbothtrialstatescombined}
p^2  + \lim_{\kappa\to\infty} \left( \Vhat \ast (\gammakappa \pm \alphakappa)(p) -  \intq \Vhatq (\gammakappaq + \alphakappaq) \right) \geq \frac{\tilde{\rho}_0-\tilde{\rho}_\gamma}{\tilde{\rho}_0} P_1^2 > 0
\end{equation}
for $P_1\leq\abs{p}\leq P_0$. Since $\tilde{\rho}_0\Vhat$ is positive, this means that there is $\epsilon>0$ (possibly smaller than the one in the previous step) such that 
\begin{equation}
a(p)^2-b(p)^2 \geq \epsilon
\end{equation}
for $\abs{p}\leq P_0$.

Step 7: Thanks to the uniform convergence, we can conclude that for all $\abs{p}\leq P_0$ and for $\kappa$ large enough
\begin{equation}
A_\kappa(p)^2- B_\kappa(p)^2 \geq \frac{\epsilon}{2}.
\end{equation}
In combination with \eqref{eqinequalitygammaAB} and the uniform boundedness of $B_\kappa$, this yields the uniform bound for $\abs{p}\leq P_0$:
\begin{equation}
\frac{1}{\alphakappap^2} \geq\frac{C}{B_\kappa(p)^2}\Rightarrow \gammakappap \leq \gammakappap +\gammakappap^2 = \alphakappap^2 \leq C.
\end{equation}
\end{proof}

\section{Conclusion}
We use the restricted minimizers for obtaining a minimizer of the full problem.
\begin{lemma}[Minimizing sequence]
\label{lemmaminimizingsequenceofrestricted}
The minimizers of the restricted problem form a minimizing sequence for the full problem:
\begin{equation}
\inf_{(\gamma,\alpha,\rhozero)\in \mathcal{D}}\functional = \lim_{\kappa \to\infty} \functionalflex{\gammakappa}{\alphakappa}{\rhozerokappa}
\end{equation}
\end{lemma}
\begin{proof}
We recall the proof of Lemma 4.9 in \cite{bogoliubovenergy}. For any $\triple \in \mathcal{D}$, we have 
\begin{equation}
(\gamma\mathbbm{1}_{\{p\mid\gamma(p)\leq\kappa\}}, \alpha \mathbbm{1}_{\{p\mid\gamma(p)\leq\kappa\}}, \rhozero) \in \mathcal{D}_\kappa.    
\end{equation}
By dominated convergence, the energies of the states containing the cutoff converge to the energy of the original state for $\kappa\to \infty$:
\begin{equation}
\functional = \limsup_{\kappa\to\infty} \functionalflex{\gamma\mathbbm{1}_{\{p\mid\gamma(p)\leq\kappa\}}}{\alpha \mathbbm{1}_{\{p\mid\gamma(p)\leq\kappa\}}}{\rhozero}\geq \limsup_{\kappa \to\infty} \functionalflex{\gammakappa}{\alphakappa}{\rhozerokappa}
\end{equation}
where the inequality follows from the fact that $\triplekappa$ is a minimizer of the restricted problem.
Taking an infimum over $\mathcal{D}$ yields
\begin{equation}
\inf_{(\gamma,\alpha,\rhozero)\in \mathcal{D}}\functional \geq \limsup_{\kappa \to\infty} \functionalflex{\gammakappa}{\alphakappa}{\rhozerokappa}.
\end{equation}
Since the reverse direction follows from the fact that the restricted domains are contained in the full one, this concludes the proof.
\end{proof}

The uniform bound tells us that the full minimization problem can be reduced to a restricted one by choosing the cutoff large enough, which shows the existence of minimizers of the full problem as follows.
\begin{proof}[Proof of Theorem \ref{theoremExistence}]
First, recall that the case $\mu\leq 0$ is trivial as stated in the beginning of the proof since the vacuum is then a minimizer. For $\mu > 0$, we can use the statements from the previous sections.
Choose $\kappa_1$ from the subsequence considered before large enough such that $\kappa_1> C(1+P_0^{-4})$, where $C$ is the constant in the uniform bound in Theorem \ref{theoremBoundrestricted}. Thus, for arbitrary elements of this subsequence, we have $\triplekappa\in \mathcal{D}_{\kappa_1}$. Combining this with Lemma \ref{lemmaminimizingsequenceofrestricted} and the observation that the energies form a decreasing sequence in $\kappa$ since the restricted domains form an increasing sequence, we get
\begin{equation}
\begin{split}
\inf_{\triple\in\mathcal{D}}\functional
&= \lim_{\kappa\to\infty}\functionalflex{\gammakappa}{\alphakappa}{\rhozerokappa}\\
&=\inf_{\kappa}\functionalflex{\gammakappa}{\alphakappa}{\rhozerokappa}\\
&\geq\inf_{\triple\in\mathcal{D}_{\kappa_1}}\functional\\
&=\functionalflex{\gamma_{\kappa_1}}{\alpha_{\kappa_1}}{\rho_{0,\kappa_1}}\\
&\geq\inf_{\triple\in\mathcal{D}}\functional
\end{split}
\end{equation}
implying that $(\gamma_{\kappa_1},\alpha_{\kappa_1},\rho_{0,\kappa_1})$ is a minimizer.
\end{proof}

The minimizers constructed in this way clearly satisfy the properties in Corollary \ref{corollaryproperties}. In fact, they also hold for any other minimizer.
\begin{proof}[Proof of Corollary \ref{corollaryproperties}]
For the first property, we can apply the argument used in Theorem \ref{theorempropertiesrestricted} (in which the same claim was proved for the restricted minimizers) to the minimizers of the full problem. This was already shown in \cite{bogoliubovenergy}. Similarly, the second property is correct because equation \eqref{eqinequalityrhozerorhogamma} holds without limits for the minimizer. By Lemma \ref{lemmaupperboundenergy}, we have $\rho_\gamma >0$ as $\functionalflex{0}{0}{\rhozero}\geq -\mu^2/2\Vhat(0)$, and therefore, the desired strict inequality holds due to \eqref{eqpositivitykineticconvolution}. The last statement follows from the proof of Theorem \ref{theoremBoundrestricted}: The Euler-Lagrange equation \eqref{eqcombinedeulerlagrange} has to hold with equality for any minimizer of the full problem. The subsequent arguments can be applied without the limits to show this bound.
\end{proof}

\begin{remark}
It is an open question whether the case of positive temperature can be treated similarly. The existence of restricted minimizers and their $p^4$ decay for $\abs{p}\geq P_0$ was already shown in \cite{bogoliubovenergy} for a different cutoff, and the proof can be adapted to the one chosen here. Therefore, only the bound for small $\abs{p}$ is left to prove. If we define the functions $A_\kappa$ and $B_\kappa$ as before, we can again obtain uniform convergence on the ball $\{\abs{p}\leq P_0\}$ to functions $a$ and $b$ by the Arzelà-Ascoli theorem. In fact, a bound of the form $\gammakappap\leq C$ for $\abs{p}\leq P_0$ would follow from the bound $a\pm b \geq \epsilon >0$ on this ball as can be seen from the Euler-Lagrange inequalities. However, the energy comparison in step 6 of the proof of the bound for $T=0$ does not work anymore since there will be no condensate for large $T$.

For general $T$ it is not clear whether such a $L^\infty$ bound holds, but this is not essential for the proof: we need the bound on $\gammakappa$ only to show that the minimizing sequence $(\gammakappa, \alphakappa, \rhozerokappa)$ is contained in a restricted domain for which there is a minimizer. Thus, if the approach yields a bound $\gammakappa\leq f$ for some $f$ and there is a minimizer on the domain with cutoff $f$, the proof is complete. For instance, $f=C/p^2$ would suffice thanks to \cite{bogoliubovenergy}.
\end{remark}

\begin{remark}
Another question to ask is whether the model can be extended to more singular potentials such as the Coulomb interaction. It is evident by the explicit appearance of $\Vhat (0)$ in the functional \eqref{eqfunctional} that this is not possible for the approach considered here since $\Vhat (0)$ is infinite for a Coulomb potential. Thus, the model has to be modified in order to remove the divergent energy contribution. One potential route for doing so could be to start with the Jellium model for bosons in the finite boxes. Here, a constant background charge density makes the system neutral, which is linked to setting $\Vhat (0)=0$. However, this changes the model significantly, and it remains open how much of the analysis can be adapted to this setting.
\end{remark}

\noindent{\bf Acknowledgments.}
This work was a master's thesis project in the program \emph{Theoretical and Mathematical Physics} at LMU and TU Munich. The author would like to express his deepest gratitude to Prof. Phan Th\`anh Nam for his guidance and support.
The author acknowledges the support from the Deutsche Forschungsgemeinschaft (DFG project Nr. 426365943).

\section*{Data Availability}
Data sharing is not applicable to this article as no new data were created or analyzed in this study.

\printbibliography
\vspace{1cm}
\end{document}